\documentclass[a4paper,USenglish]{lipics-v2018}

\usepackage{microtype}
\usepackage[T1]{fontenc}
\usepackage[utf8]{inputenc}
\usepackage{dsfont, complexity}
\usepackage{IEEEtrantools}
\usepackage{lmodern}
\usepackage{float}
\usepackage{graphicx, hyperref}
\usepackage{xcolor,colortbl}
\usepackage{paralist,tabularx}
\usepackage[ruled]{algorithm} 
\usepackage{algpseudocode,algorithmicx}
\usepackage{comment}
\usepackage{thmtools, mathtools}
\usepackage{thm-restate}
\usepackage{soul}

\usepackage{tikz}
\usetikzlibrary{arrows,decorations.markings,patterns,positioning,shapes,decorations.pathreplacing,calc}
\tikzset{draw half paths/.style 2 args={%
  decoration={show path construction,
    lineto code={
      \draw [#1] (\tikzinputsegmentfirst) -- 
         ($(\tikzinputsegmentfirst)!0.5!(\tikzinputsegmentlast)$);
      \draw [#2] ($(\tikzinputsegmentfirst)!0.5!(\tikzinputsegmentlast)$)
        -- (\tikzinputsegmentlast);
    }
  }, decorate
}}

\newcommand{\algnin}{\hspace*{-\algorithmicindent}}

\tikzstyle{vertex} = [circle, draw=black, fill=black, scale= 0.5]
\tikzstyle{edgelabel} = [rectangle, fill=white]

\newtheorem{claim}{Claim}
\newtheorem{ex}[theorem]{Example}

\newtheorem{cor}[claim]{Corollary}

\newtheorem{defi}[theorem]{Definition}

\usepackage[textsize=tiny,textwidth=2.2cm,shadow]{todonotes}

\definecolor{MyColor}{RGB}{197,0,205}
\definecolor{MyPurple}{RGB}{197,0,205}
\pgfmathsetmacro{\d}{2}
\pgfmathsetmacro{\b}{2}

\newcommand{\I}{\mathcal{I}}
\newcommand{\df}{\delta}
\newcommand{\Ord}[1]{\mathcal{O}\left(#1\right)}
\newcommand{\Om}[1]{\Omega\left(#1\right)}

\nolinenumbers

\title{Pairwise preferences in the stable marriage problem}

\author{\'{A}gnes Cseh}{Institute of Economics, Centre for Economic and Regional Studies, Hungarian Academy of Sciences, 1097 Budapest, T\'{o}th K\'{a}lm\'{a}n u. 4., Hungary}{cseh.agnes@krtk.mta.hu}{
}{}

\author{Attila Juhos}{Department of Computer Science and Information Theory, Budapest University of Technology and Economics, 1117 Budapest, Magyar Tud\'{o}sok krt. 2., Hungary}{juhosattila@cs.bme.hu}{}{}

\authorrunning{\'{A}. Cseh and A. Juhos}

\Copyright{\'{A}gnes Cseh and Attila Juhos}

\subjclass{Theory of computation $\rightarrow$ Design and analysis of algorithms $\rightarrow$ Graph algorithms analysis}

\keywords{stable marriage, intransitivity, acyclic preferences, poset, weakly stable matching, strongly stable matching, super stable matching}

\relatedversion{}

\funding{This work was supported by the Cooperation of Excellences Grant (KEP-6/2017), by the Ministry of Human Resources under its New National Excellence Programme (ÚNKP-18-4-BME-331 and ÚNKP-18-1-I-BME-309), the Hungarian Academy of Sciences under its Momentum Programme (LP2016-3/2016), its J\'anos Bolyai Research Fellowship, and OTKA grant K128611.}

\acknowledgements{The authors thank Tam\'{a}s Fleiner, David Manlove, and D\'{a}vid Szeszl\'{e}r for fruitful discussions on the topic.}

\begin{document}
\pagenumbering{gobble}
\clearpage
\thispagestyle{empty}

\maketitle

\begin{abstract}
We study the classical, two-sided stable marriage problem under pairwise preferences. In the most general setting, agents are allowed to express their preferences as comparisons of any two of their edges and they also have the right to declare a draw or even withdraw from such a comparison. This freedom is then gradually restricted as we specify six stages of orderedness in the preferences, ending with the classical case of strictly ordered lists. We study all cases occurring when combining the three known notions of stability---weak, strong and super-stability---under the assumption that each side of the bipartite market obtains one of the six degrees of orderedness. By designing three polynomial algorithms and two $\NP$-completeness proofs we determine the complexity of all cases not yet known, and thus give an exact boundary in terms of preference structure between tractable and intractable cases.
\end{abstract}
 
\newpage
\clearpage
\pagenumbering{arabic}
\section{Introduction}

In the 2016 USA Presidential Elections, polls unequivocally reported Democratic presidential nominee Bernie Sanders to be more popular than Republican candidate Donald Trump~\cite{WWWHuffPost,WWWUSelection}. However, Sanders was beaten by Clinton in their own party's primary election cycle, thus the 2016 Democratic National Convention endorsed Hillary Clinton to be the Democrat's candidate. In the Presidential Elections, Trump defeated Clinton. This recent example demonstrates well how inconsistent pairwise preferences can be.

Preferences play an essential role in the stable marriage problem and its extensions. In the classical setting~\cite{GS62}, each man and woman expresses their preferences on the members of the opposite gender by providing a strictly ordered list. A set of marriages is stable if no pair of agents blocks it. A man and woman form a blocking pair if they mutually prefer one another to their respective spouses.

Requiring strict preference orders in the stable marriage problem is a strong assumption, which rarely suits real world scenarios~\cite{Bir17}. The study of less restrictive preference structures has been flourishing~\cite{ABFGHMR17,FGK16, Irv94,IMS03,KMMP07, Man02} for decades. As soon as one allows for ties in preference lists, the definition of a blocking edge needs to be revisited. In the literature, three intuitive definitions are used, each  of which defines weakly, strongly and super stable matchings. According to weak stability, a matching is blocked by an edge $uw$ if agents $u$ and $w$ both strictly prefer one another to their partners in the matching. A strongly blocking edge is preferred strictly by one end vertex, whereas it is not strictly worse than the matching edge at the other end vertex. A blocking edge is at least as good as the matching edge for both end vertices in the super stable case. Super stable matchings are strongly stable and strongly stable matchings are weakly stable by definition. 

Weak stability is an intuitive notion that is most aligned with the classical blocking edge definition in the model defined by Gale and Shapley~\cite{GS62}. However, reaching strong stability is the goal to achieve in many applications, such as college admission programs. In most countries, students need to submit a strict ordering in the application procedure, but colleges are not able to rank all applicants strictly, hence large ties occur in their lists. According to the equal treatment policy used in Chile and Hungary for example, it may not occur that a student is rejected from a college preferred by him, even though other students with the same score are admitted~\cite{BK15, RLPC14}. Other countries, such as Ireland~\cite{Chen2012mip}, break ties with lottery, which gives way to a weakly stable solution. Super stable matchings are admittedly less relevant in applications, however, they represent worst-case scenarios if uncertain information is given about the agents' preferences. If two edges are incomparable to each other due to incomplete information derived from the agent, then it is exactly the notion of a super stable matching that guarantees stability, no matter what the agent's true preferences are.

The goal of our present work is to investigate the three cases of stability in the presence of more general preference structures than ties.

\subsection{Related work}

The study of cyclic and intransitive preferences has been triggering scientists from a wide range of fields for decades. Blavatsky~\cite{Bla03} demonstrated that in choice situations under risk, the overwhelming majority of individuals expresses intransitive choice and violation of standard consistency requirements. Humphrey~\cite{Hum01} found that cyclic preferences persist even when the choice triple is repeated for the second time. Using MRI scanners, neuroscientists identified brain regions encoding `local desirability', which led to clear, systematic and predictable intransitive choices of the participants of the experiment~\cite{KTHDP10}.

Cyclic and intransitive preferences occur naturally in multi-attribute comparisons~\cite{Fis99,May54}. May~\cite{May54} studied the choice on a prospective partner and found that a significant portion of the participants expressed the same cyclic preference relations if candidates lacking exactly one of the three properties intelligence, looks, and wealth were offered at pairwise comparisons. Cyclic and intransitive preferences also often emerge in the broad topic of voting and representation, if the set of voters differs for some pairwise comparisons~\cite{AR96}, such as in our earlier example with the polls on the Clinton--Sanders--Trump battle. Preference aggregation is another field that often yields intransitive group preferences, as the famous Condorcet-paradox~\cite{Con85} also states. In this paper, we investigate the stable marriage problem equipped with these ubiquitous and well-studied preference structures.

Regarding the stable marriage problem, all three notions of stability have been thoroughly investigated if preferences are given in the form of a partially ordered set, a list with ties or a strict list~\cite{GS62, Irv94,IMS03,KMMP07,Man02,Man13}. Weakly stable matchings always exist and can be found in polynomial time~\cite{Man02}, and a super stable matching or a proof for its non-existence can also be produced in polynomial time~\cite{Irv94,Man13}. The most sophisticated ideas are needed in the case of strong stability, which turned out to be solvable in polynomial time if both sides have tied preferences~\cite{Irv94}. Irving~\cite{Irv94} remarked that ``Algorithms that we have described can easily be extended to the more general problem in which each person's preferences are expressed as a partial order. This merely involves interpreting the `head' of each person's (current) poset as the set of source nodes, and the `tail' as the set of sink nodes, in the corresponding directed acyclic graph.'' Together with his coauthors, he refuted this statement for strongly stable matchings and shows that exchanging ties for posets actually makes the strongly stable marriage problem $\NP$-complete~\cite{IMS03}. We show it in this paper that the intermediate case, namely when one side has ties preferences, while the other side has posets, is solvable in polynomial time.

Beyond posets, studies on the stable marriage problem with general preferences occur sporadically. These we include in Table~\ref{ta:all} to give a structured overview on them. Intransitive, acyclic preference lists were permitted by Abraham~\cite{Abr03}, who connects the stable roommates problem with the maximum size weakly stable marriage problem with intransitive, acyclic preference lists in order to derive a structural perspective. Aziz et al.~\cite{ABFGHMR17} discussed the stable marriage problem under uncertain pairwise preferences. They also considered the case of certain, but cyclic preferences and show that deciding whether a weakly stable matching exists is $\NP$-complete if both sides can have cycles in their preferences. Strongly and super stable matchings were discussed by Farczadi et al.~\cite{FGK16}. Throughout their paper they assumed that one side has strict preferences, and show that finding a strongly or a super stable matching (or proving that none exists) can be done polynomial time if the other side has cyclic lists, where cycles of length at least~3 are permitted to occur, but the problems become $\NP$-complete as soon as cycles of length~2 are also allowed.

\begin{table}[htbp]
	\centering
      \resizebox{\textwidth}{!}{
		\begin{tabular}{|l|c c c c c|}
		\hline
			WEAK& strict & ties & poset & acyclic & asymmetric or arbitrary\\ \hline
            strict & $\mathcal{O}(m)$~\cite{GS62} & $\mathcal{O}(m)$~\cite{Irv94}& $\mathcal{O}(m)$~\cite{Man02} & \cellcolor{MyColor!25}$\mathcal{O}(m)$ & \cellcolor{MyColor!25}\textbf{$\NP$}\\
            ties & \phantom{n} & $\mathcal{O}(m)$~\cite{Irv94}& $\mathcal{O}(m)$~\cite{Man02} & \cellcolor{MyColor!25}$\mathcal{O}(m)$ & \cellcolor{MyColor!25}\textbf{$\NP$}\\
            poset & \phantom{n} & \phantom{n} & $\mathcal{O}(m)$~\cite{Man02} & \cellcolor{MyColor!25}$\mathcal{O}(m)$ & \cellcolor{MyColor!25}\textbf{$\NP$}\\
            acyclic & \phantom{n} & \phantom{n} & \phantom{n} & \cellcolor{MyColor!25}$\mathcal{O}(m)$ & \cellcolor{MyColor!25}\textbf{$\NP$}\\
            asymmetric or arbitrary & \phantom{n} & \phantom{n} & \phantom{n} & \phantom{n} & $\NP$~\cite{ABFGHMR17}\\
		\hline
		\end{tabular}
        }
        \vspace{5mm}
        
       \resizebox{\textwidth}{!}{
		\begin{tabular}{|l|c c c c c c|}
		\hline
			STRONG & strict & ties & poset & acyclic & asymmetric & arbitrary\\ \hline
             strict & $\mathcal{O}(m)$~\cite{GS62} & $\mathcal{O}(nm)$~\cite{Irv94,KMMP07} & pol~\cite{FGK16} & pol~\cite{FGK16} & pol~\cite{FGK16} & $\NP$~\cite{FGK16}\\
            ties & \phantom{n} & $\mathcal{O}(nm)$~\cite{Irv94,KMMP07} & \cellcolor{MyColor!25}$\Ord{mn^2 + m^2}$ & \cellcolor{MyColor!25}$\Ord{mn^2 + m^2}$ &  \cellcolor{MyColor!25}$\Ord{mn^2 + m^2}$ &  $\NP$~\cite{FGK16}\\
            poset & \phantom{n} & \phantom{n} & $\NP$~\cite{IMS03} & $\NP$~\cite{IMS03} & $\NP$~\cite{IMS03} & $\NP$~\cite{IMS03}\\
            acyclic & \phantom{n} & \phantom{n} & \phantom{n} & $\NP$~\cite{IMS03} & $\NP$~\cite{IMS03} & $\NP$~\cite{IMS03}\\
            asymmetric & \phantom{n} & \phantom{n} & \phantom{n} & \phantom{n} & $\NP$~\cite{IMS03} & $\NP$~\cite{IMS03}\\
            arbitrary & \phantom{n} & \phantom{n} & \phantom{n} & \phantom{n} & \phantom{n} & $\NP$~\cite{IMS03}\\
		\hline
		\end{tabular}
        } 
        \vspace{5mm}
        
        \resizebox{1\textwidth}{!}{
		\begin{tabular}{|l|c c c c c c|}
		\hline
			SUPER & strict & ties & poset & acyclic & asymm. & arbitrary\\ \hline
            strict & $\mathcal{O}(m)$~\cite{GS62} & $\mathcal{O}(m)$~\cite{Irv94} & $\mathcal{O}(m)$~\cite{Irv94, Man13} & $\mathcal{O}(m)$~\cite{FGK16} & $\mathcal{O}(m)$~\cite{FGK16} & $\NP$~\cite{FGK16}\\
            ties & \phantom{n}  & $\mathcal{O}(m)$~\cite{Irv94} & $\mathcal{O}(m)$~\cite{Irv94, Man13} & \cellcolor{MyColor!25}\textbf{$\Ord{n^2m}$} & \cellcolor{MyColor!25}\textbf{$\Ord{n^2m}$} & $\NP$~\cite{FGK16}\\
            poset & \phantom{n} & \phantom{n} & $\mathcal{O}(m)$~\cite{Irv94, Man13} & \cellcolor{MyColor!25}\textbf{$\Ord{n^2m}$} & \cellcolor{MyColor!25}\textbf{$\Ord{n^2m}$} & $\NP$~\cite{FGK16}\\
            acyclic & \phantom{n} & \phantom{n} & \phantom{n} & \cellcolor{MyColor!25} $\NP$  & \cellcolor{MyColor!25} $\NP$ & $\NP$~\cite{FGK16}\\
            asymmetric & \phantom{n} & \phantom{n} & \phantom{n} & \phantom{n} & \cellcolor{MyColor!25}\textbf{$\NP$} & $\NP$~\cite{FGK16}\\
            arbitrary & \phantom{n} & \phantom{n} & \phantom{n} & \phantom{n} & \phantom{n} & $\NP$~\cite{FGK16}\\
		\hline
		\end{tabular}
       }
        
\caption{The complexity tables for weak, strong and super-stability.}
\label{ta:all}
\end{table}

\subsection{Our contribution}

This paper aims to provide a coherent framework for the complexity of the stable marriage problem under various preference structures. We consider the three known notions of stability: weak, strong and super. In our analysis we distinguish six stages of entropy in the preference lists; strict lists, lists with ties, posets, acyclic pairwise preferences, asymmetric pairwise preferences and arbitrary pairwise preferences. All of these have been defined in earlier papers, along with some results on them. Here we collect and organize these known results in all three notions of stability, considering six cases of orderedness for each side of the bipartite graph. Table~\ref{ta:all} summarizes these results. 

Each of the three tables contained empty cells, this is, cases with unknown complexity so far. These are denoted by color in Table~\ref{ta:all}. We fill all gaps, providing two $\NP$-completeness proofs and three polynomial time algorithms. Interestingly, the three tables have the border between polynomial time and $\NP$-complete cases at very different places.

\textbf{Structure of the paper.} We define the problem variants formally in Section~\ref{se:prel}. Weak, strong and super stable matchings are then discussed in Sections~\ref{se:weak}, \ref{se:strong} and \ref{se:super}, respectively.

\section{Preliminaries}
\label{se:prel}

In the stable marriage problem, we are given a not necessarily complete bipartite graph $G = (U \cup W, E)$, where vertices in $U$ represent men, vertices in $W$ represent women, and edges mark the acceptable relationships between them. Each person $v \in U \cup W$ specifies a set $\mathcal{R}_v$ of pairwise comparisons on the vertices adjacent to them. These comparisons as ordered pairs define four possible relations between two vertices $a$ and $b$ in the neighborhood of~$v$.
\begin{itemize}
\item $a$ is preferred to $b$, while $b$ is not preferred to $a$ by $v$: $a \prec_v b$;
\item $a$ is not preferred to $b$, while $b$ is preferred to $a$ by $v$: $a \succ_v b$;
\item $a$ is not preferred to $b$, neither is $b$ preferred to $a$ by $v$: $a \sim_v b$;
\item $a$ is preferred to $b$, so is $b$ preferred to $a$ by $v$: $a ||_v b$.
\end{itemize}

In words, the first two relationships express that an agent $v$ \emph{prefers} one agent \emph{strictly} to the other. The third option is interpreted as \emph{incomparability}, or a not yet known relation between the two agents. The last relation tells that $v$ knows for sure that the two options are \emph{equally good}. For example, if $v$ is a sports sponsor considering to offer a contract to exactly one of players $a$ and $b$, then $v$'s preferences are described by these four relations in the following scenarios: $a$ beats $b$, $b$ beats $a$, $a$ and $b$ have not played against each other yet, and finally, $a$ and $b$ played a draw.

We say that edge $va$ \emph{dominates} edge $vb$ if $a \prec_v b$. If $a \prec_v b$ or $a \sim_v b$, then $b$ is \emph{not preferred to} $a$. The partner of vertex $v$ in matching $M$ is denoted by~$M(v)$. The neighborhood  of $v$ in graph $G$ is denoted by $\mathcal{N}_G(v)$ and it consists of all vertices that are adjacent to $v$ in~$G$. To ease notation, we introduce the empty set as a possible partner to each vertex, symbolizing the vertex remaining unmatched in a matching~$M$ ($M(v)=\emptyset$). As usual, being matched to any acceptable vertex is preferred to not being matched at all: $a \prec_v \emptyset$ for every $a \in \mathcal{N}(v)$. Edges to unacceptable partners do not exist, thus these are not in any pairwise relation to each other or to edges incident to~$v$.

We differentiate six degrees of preference orderedness in our study.
\begin{enumerate}
\item The strictest, classical two-sided model~\cite{GS62} requires each vertex to rank all of its neighbors in a \emph{strict} order of preference. For each vertex, this translates to a transitive and complete set of pairwise relations on all adjacent vertices.
\item This model has been relaxed very early to lists admitting \emph{ties}~\cite{Irv94}. The pairwise preferences of vertex $v$ form a preference list with ties if the neighbors of $v$ can be clustered into some sets $N_1, N_2, \dots, N_k$ so that vertices in the same set are incomparable, while for any two vertices in different sets, the vertex in the set with the lower index is strictly preferred to the other one.
\item Following the traditions~\cite{FIM07,IM02,IMS03,Man02}, the third degree of orderedness we define is when preferences are expressed as \emph{posets}. Any set of antisymmetric and transitive pairwise preferences by definition forms a partially ordered set.
\item By dropping transitivity but still keeping the structure cycle-free, we arrive to \emph{acyclic} preferences~\cite{Abr03}. This category allows for example $a \sim_v c$ , if $a \prec_v b \prec_v c$, but it excludes $a ||_v c$ and $a \succ_v c$.
\item \emph{Asymmetric} preferences~\cite{FGK16} may contain cycles of length at least~3. This is equivalent to dropping acyclicity from the previous cluster, but still prohibiting the indifference relation $a||_v b$, which is essentially a 2-cycle in the form $a$ is preferred to $b$, and $b$ is preferred to~$a$.
\item Finally, an \emph{arbitrary} set of pairwise preferences can also be allowed~\cite{ABFGHMR17,FGK16}.
\end{enumerate}

A matching is \emph{stable} if it admits no blocking edge. For strict preferences, a blocking edge was defined in the seminal paper of Gale and Shapley~\cite{GS62}: an edge $uv \notin M$ blocks matching $M$ if both $u$ and $v$ prefer each other to their partner in $M$ or they are unmatched. Already when extending this notion to preference lists with ties, one needs to specify how to deal with incomparability. Irving~\cite{Irv94} defined three notions of stability. We extend them to pairwise preferences in the coming three sections. We omit the adjectives weakly, strongly, and super wherever there is no ambiguity about the type of stability in question. All missing proofs can be found in the \hyperlink{app:appendix}{Appendix}.

\section{Weak stability}

In weak stability, an edge outside of $M$ blocks $M$ if it is \emph{strictly preferred} to the matching edge by \emph{both} of its end vertices. From this definition follows that $w ||_u w'$ and $w \sim_u w'$ are exchangeable in weak stability, because blocking occurs only if the non-matching edge dominates the matching edges at both end vertices. Therefore, an instance with arbitrary pairwise preferences can be assumed to be asymmetric.

\label{se:weak}
\begin{defi}[blocking edge for weak stability]
Edge $uw$ blocks $M$, if 
\begin{enumerate}
	\item $uw \notin M$;
	\item $w \prec_u M(u)$;
	\item $u \prec_w M(w)$.
\end{enumerate}
\end{defi}

For weak stability, preference structures up to posets have been investigated, see Table~\ref{ta:all}. A stable solution is guaranteed to exist in these cases~\cite{Irv94,Man02}. Here we extend this result to acyclic lists, and complement it with a hardness proof for all cases where asymmetric lists appear, even if they do so on one side only.

\begin{theorem}
\label{th:smsas_pol}
Any instance of the stable marriage problem with acyclic pairwise preferences for all vertices admits a weakly stable matching, and there is a polynomial time algorithm to determine such a matching.
\end{theorem}

\begin{proof}
    We utilize a widely used argument~\cite{Irv94} to show this. For acyclic relations $\mathcal{R}_v$, a linear extension $\mathcal{R}'_v$ of $\mathcal{R}_v$ exists. The extended instance with linear preferences is guaranteed to admit a stable matching~\cite{GS62}. Compared to $\mathcal{R}$, relations in $\mathcal{R}'_v$ impose more constraints on stability, therefore, they can only restrict the original set of weakly stable solutions. If both sides have acyclic lists, a stable matching is thus guaranteed to exist and a single run of the Gale-Shapley algorithm on the extended instance delivers one.
\end{proof}

Stable matchings are not guaranteed to exist as soon as a cycle appears in the preferences, as Example~\ref{ex:weak_gadget} demonstrates. Theorem~\ref{th:weak_np} shows that the decision problem is in fact hard from that point on.

\begin{ex}
\label{ex:weak_gadget}
No stable matching can be found in the following instance with strict lists on one side and asymmetric lists on the other side. There are three men $u_1, u_2, u_3$ adjacent to one woman $w$. The woman's pairwise preferences are cyclic: $u_1 \prec u_2, u_2 \prec u_3, u_3\prec u_1$. Any stable matching $M$ must consist of a single edge. Since the men's preferences are identical, we can assume that $u_1w \in M$ without loss of generality. Then $u_3w$ blocks~$M$.
\end{ex}

\begin{restatable}{theorem}{thweaknp}
\label{th:weak_np}
If one side has strict lists, while the other side has asymmetric pairwise preferences, then determining whether a weakly stable matching exists is $\NP$-complete, even if each agent finds at most four other agents acceptable.
\end{restatable}

\section{Strong stability}
\label{se:strong}

In strong stability, an edge outside of $M$ blocks $M$ if it is \emph{strictly preferred} to the matching edge by \emph{one} of its end vertices, while the other end vertex \emph{does not prefer} its matching edge to it.

\begin{defi}[blocking edge for strong stability]
Edge $uw$ blocks $M$, if\\ 
\begin{minipage}{0.4\textwidth}
\begin{enumerate}
	\item $uw \notin M$;
	\item $w \prec_u M(u)$ or $ w \sim_u M(u)$;
	\item $u \prec_w M(w)$,
\end{enumerate}
\end{minipage}
or \hspace{10mm}
\begin{minipage}{0.4\textwidth}
\begin{enumerate}
	\item $uw \notin M$;
	\item $w \prec_u M(u)$;
	\item $u \prec_w M(w)$ or $u \sim_w M(w)$.
\end{enumerate}
\end{minipage}
\end{defi}

The largest set of relevant publications has appeared on strong stability, yet gaps were present in the complexity table, see Table~\ref{ta:all}. In this section we present a polynomial algorithm that is valid in all cases not solved yet. We assume men to have preference lists with ties, and women to have asymmetric relations. Our algorithm returns a strongly stable matching or a proof for its nonexistence. It can be seen as an extended version of Irving's algorithm for strongly stable matchings in instances with ties on both sides~\cite{Irv94}. Our contribution is a sophisticated rejection routine, which is necessary here, because of the intransitivity of preferences on the women's side. The algorithm in~\cite{FGK16} solves the problem for strict lists on the men's side, and it is much simpler than ours. It was designed for super stable matchings, but strong and super stability do not differ if one side has strict lists. For this reason, that algorithm is not suitable for an extension in strong stability.

Roughly speaking, our algorithm alternates between two phases, both of which iteratively eliminate edges that cannot occur in a strongly stable matching. In the first phase, Gale-Shapley proposals and rejections happen, while the second phase focuses on finding a vertex set violating the Hall condition in a specified subgraph. Finally, if no edge can be eliminated any more, then we show that an arbitrary maximum matching is either stable or it is a proof for the non-existence of stable matchings. Algorithms~\ref{alg:strong_SMTAs} and~\ref{alg:reject} below provide a pseudocode. The time complexity analysis has been shifted to the Appendix.

The second phase of the algorithm relies on the notion of the \emph{critical set} in a bipartite graph, also utilized in~\cite{Irv94}, which we sketch here. For an exhaustive description we refer the reader to~\cite{LP09}. The well-known Hall-condition~\cite{Hal35} states that there is a matching covering the entire vertex set $U$ if and only if for each $X \subseteq U$, $|\mathcal{N}(X)| \geq |X|$. Informally speaking, the reason for no matching being able to cover all the vertices in $U$ is that a subset $X$ of them has too few neighbors in $W$ to cover their needs. The difference $\df(X) = |X| - |\mathcal{N}(X)|$ is called the \emph{deficiency of $X$}. It is straightforward that for any $X \subseteq U$, at least $\df(X)$ vertices in $X$ cannot be covered by any matching in $G$, if $\df(X) > 0$. Let $\df(G)$ denote the maximum deficiency over all subsets of $U$. Since $\df(\emptyset) = 0$, we know that $\df(G) \geq 0$. Moreover, it can be shown the size of maximum matching is $\nu(G) = |U| - \df(G)$. If we let $Z_1, Z_2$ be two arbitrary subsets of $U$ realizing the maximum deficiency, then $Z_1 \cap Z_2$ has maximum deficiency as well. Therefore, the intersection of all maximum-deficiency subsets of $U$ is the unique set with maximum deficiency with the following properties: \begin{inparaenum}
\item it has the lowest number of elements and \item it is contained in all other subsets with maximum deficiency.
\end{inparaenum} This set is called the \emph{critical set} of~$G$. Last but not least, it is computationally easy to determine the critical set, since for any maximum matching $M$ in $G$, the critical set consists of vertices in $U$ not covered by $M$ and vertices in $U$ reachable from the uncovered ones via an alternating path.

\begin{theorem}
\label{th:st_pol}
If one side has tied preferences, while the other side has asymmetric pairwise preferences, then deciding whether the instance admits a strongly stable matching can be done in $\mathcal{O}(mn^2+m^2)$ time.
\end{theorem}


\newcounter{myalgcounter}

\begin{algorithm}
	\caption{Strongly stable matching with ties and asymmetric relations}
	\textbf{Input}: $\I=(U, W, E, \mathcal{R}_{U}, \mathcal{R}_{W})$; $\mathcal{R}_{U}$: lists with ties, $\mathcal{R}_{W}$: asymmetric.
	
	\begin{algorithmic}[1] %
        
		\Statex \algnin{} \textbf{INITIALIZATION} %
        
		\State for each $u \in U$ add an extra woman $w_u$ at the end of his list; $w_u$ is only acceptable for~$u$ \label{li:dummy} 
		
        \State set all edges to be inactive \label{li:inact} \medskip
        
		\Statex \algnin{} \textbf{PHASE 1} %
		
        \While{there exists a man with no active edge} \label{li:ph1start} %
        	  \State propose along all edges of each such man $u$ in the next tie on his list \label{li:nexttie} 
            
            \For{each new proposal edge $uw$} \label{li:newactive}
            	\State reject all edges $u'w$ such that $u \prec_{w} u'$ \label{li:rejectworse} 
            \EndFor \label{li:9}
                   
			\State \label{li:ph1reject} STRONG\_REJECT() %
            
		\EndWhile \label{li:ph1end}		    \medskip
        
		\Statex \algnin{} \textbf{PHASE 2} %
		
		\State let $G_A$ be the graph of active edges with $V(G_A) = U \cup W$ \label{li:GA}%

		\State let $U' \subseteq U$ be the critical set of men with respect to $G_A$ \label{li:Hall}
        
        \If{$U' \neq \emptyset$}\label{li:u'emptyset}
        	
            \State all active edges of each $u \in U'$ are rejected \label{li:Hall_reject}
		
			\State STRONG\_REJECT() \label{li:Hall_strong_reject}
		
			\State \textbf{goto PHASE 1} \label{li:goto} %
        
        \EndIf  \medskip
        
		\Statex \algnin{} \textbf{OUTPUT} %
        
        \State let M be a maximum matching in $G_A$ \label{li:M} %
        
        \If{$M$ covers all women who have ever had an active edge} \label{li:cover}
        
			\State STOP, OUTPUT $M \cap E$ and ``There is a strongly stable matching.'' \label{li:outputM} %
        
		\Else \label{li:22}
        	
            \State STOP, OUTPUT ``There is no strongly stable matching.'' \label{li:outputNO} %
		
		\EndIf \label{li:24}
		
        \setcounter{myalgcounter}{\value{ALG@line}}
	\end{algorithmic}
	\label{alg:strong_SMTAs}
\end{algorithm}

\begin{algorithm}
	
 	\caption{STRONG\_REJECT()}
	    
	\begin{algorithmic}[1]
    
    	\setcounter{ALG@line}{\value{myalgcounter}}
    	\State let $R = U$\label{li:25}

        \While{$R \neq \emptyset$} \label{li:newset} 

            \State let $u$ be an element of $R$  \label{li:27}
            
            \If{$u$ has exactly one active edge $uw$}  \label{li:oneact}
            	\State reject all $u'w$ such that $u' \sim_{w} u$ \label{li:29}
                \State if $u'w$ was active, then let $R := R \cup \{u'\}$ \label{li:30}
                
            \ElsIf{$u$ has no active edge} \label{li:zeroact}
            	\State reject all $u'w$ such that $w$ is in the proposal tie of $u$ and $u' \sim_{w} u$\label{li:34}
                \State if $u'w$ was active, then let $R := R \cup \{u'\}$ \label{li:zeroact_r}
            \EndIf 
            
            \State let $R := R \setminus \{ u \}$ 
		\EndWhile
        
        \setcounter{myalgcounter}{\value{ALG@line}}
	\end{algorithmic}
    \label{alg:reject}
\end{algorithm}


\textbf{Initialization.} For the clarity of our proofs we add a dummy partner $w_u$ to the bottom of the list of each man $u$, where $w_u$ is not acceptable to any other man (line~\ref{li:dummy}). We call the modified instance~$\I'$. This standard technical modification is to ensure that all men are matched in all stable matchings. At start, all edges are \emph{inactive} (line~\ref{li:inact}).

\textbf{First phase.} The first phase of our algorithm (lines~\ref{li:ph1start}-\ref{li:ph1end}) imitates the classical Gale-Shapley deferred acceptance procedure. In the first round, each unmatched man simultaneously proposes to all women in his top tie (line~\ref{li:nexttie}). Inactive edges that carry a proposal become \emph{active} as soon as the proposal arrives. The tie that a man has just proposed along is called the man's \emph{proposal tie}. If all edges in the proposal tie are rejected, the man steps down on his list and proposes along all edges in the next tie (lines~\ref{li:ph1start}-\ref{li:nexttie}).

Proposals cause two types of rejections in the graph (lines~\ref{li:newactive}-\ref{li:ph1reject}), based on the rules below. Notice that these rules are more sophisticated than in the Gale-Shapley or Irving algorithms~\cite{GS62,Irv94}. The most striking difference may be that rejected edges are not deleted from the graph, since they can very well carry a proposal later. However, the term active only describes proposal edges that have not been rejected yet, not even prior to the proposal.

\begin{itemize}
\item For each new proposal (but not necessarily active) edge $uw$, $w$ rejects all edges to which $uw$ is strictly preferred (lines~\ref{li:newactive}-\ref{li:9}). Note again that $uw$ might have been rejected earlier than being proposed along, in which case $uw$ is a proposal edge without being active.

\item The second kind of rejections are detailed in Algorithm~\ref{alg:reject}. We search for a man in the set $R$ of men to be investigated, whose set of active edges has cardinality at most~1 (lines~\ref{li:25}-\ref{li:27}). If any such man has exactly one active edge $uw$ (line~\ref{li:oneact}), then all other edges that are incident to $w$ and incomparable to $uw$ are rejected (line~\ref{li:29}). If man $u'$ has lost an active edge in the previous operation, then $u'$ is added back to the set $R$ of men to be investigated in later rounds (line~\ref{li:30}). The other case is when a man $u$ has no active edge at all (line~\ref{li:zeroact}). In this case, all edges that are incident to any neighbor $w$ of $u$ in his---now fully rejected---proposal tie and incomparable to $uw$ at $w$ are rejected (line~\ref{li:34}). The set $R$ is again supplemented by those men who lost active edges during the previous operation (line~\ref{li:zeroact_r}). Finally, the man $u$ chosen at the beginning of this rejection round is excluded from~$R$.

\end{itemize}
As mentioned earlier, men without any active edge proceed to propose along the next tie in their list. These operations are executed until there is no more edge to propose along or to reject, which marks the end of the first phase.

\textbf{Second phase.} In the second phase, the set of active edges induce the graph~$G_A$, on which we examine the critical set $U'$  (lines~\ref{li:GA}-\ref{li:Hall}). If $U'$ is not empty, then all active edges of each $u \in U'$ are rejected (line~\ref{li:Hall_reject}). These rejections might trigger more rejections, which is done by calling Algorithm~\ref{alg:reject} as a subroutine (line~\ref{li:Hall_strong_reject}). The mass rejections in line~\ref{li:Hall_reject} generate a new proposal tie for at least one man, returning to the first phase (line~\ref{li:goto}). Note that an empty critical set leads to producing the output, which is described just below.

\textbf{Output.}
In the final set of active edges, an arbitrary maximum matching $M$ is calculated (line~\ref{li:M}). If $M$ covers all women who have ever had an active edge, then we send it to the output (lines~\ref{li:cover}-\ref{li:outputM}), otherwise we report that no stable matching exists (lines~\ref{li:22}-\ref{li:outputNO}).

\smallskip

We prove Theorem~\ref{th:st_pol} via a number of claims, building up the proof as follows. The first three claims provide the technical footing for the last two claims. Claim~\ref{le:dummy_women} is a rather technical observation about the righteousness of the input initialization. An edge appearing in any stable matching is called a \emph{stable edge}. Claim~\ref{le:stable_edges} shows that no stable edge is ever rejected. Claim~\ref{le:women_activated} proves that all stable matchings must cover all women who have ever received an offer. Then, Claim~\ref{le:cover_is_stable} proves that if the algorithm outputs a matching, then it must be stable, and Claim~\ref{le:stable_max_all_women} shows the opposite direction; if stable matchings exist, then one is outputted by our algorithm.

\begin{claim}
	\label{le:dummy_women}
	A matching in $\I'$ is stable if and only if its restriction to $\I$ is stable and it covers all men in $\I'$.
\end{claim}

\begin{proof}
	If a matching in $\I'$ leaves a man $u$ unmatched, then $u w_u$ blocks the matching. Thus all stable matchings in $\I'$ cover all men. Furthermore, the restriction to~$\I$ of a stable matching in $\I'$ cannot be blocked by any edge in $\I$, because this blocking edge also exists in~$\I'$.
	
	A stable matching in $\I$, supplemented by the dummy edges for all unmatched men cannot be blocked by any edge in $\I'$, because dummy edges are last-choice edges and regular edges block in both instances simultaneously.
\end{proof}

\begin{claim}
	\label{le:stable_edges}
	No stable edge is ever rejected in the algorithm.
\end{claim}

\begin{proof}
    Let us suppose that $uw$ is the first rejected stable edge and the corresponding stable matching is~$M$. There are four rejection calls, in lines~\ref{li:rejectworse}, \ref{li:Hall_reject}, \ref{li:29}, and \ref{li:34}. In all cases we derive a contradiction. Our arguments are illustrated in Figure~\ref{fi:stable_edges}.
    
	\begin{itemize}
  		
        \item Line~\ref{li:rejectworse}: $uw$ was rejected because $w$ received a proposal from a man $u'$ such that $u' \prec_{w} u$.\\        
        Since $M$ is stable, $u'$ must have a partner $w'$ in $M$ such that $w' \prec_{u'} w$. We also know that $u'$ has reached $w$ with its proposal ties, thus, due to the monotonicity of proposals, $u'w' \in M$ must have been rejected before $uw$ was rejected. This contradicts our assumption that $uw$ was the first rejected stable edge.
        
        \item Lines~\ref{li:29} and \ref{li:34}: rejection was caused by a man $u'$ such that $u' \sim_{w} u$.\\
        Either the whole proposal tie of $u'$ was rejected or $u'w$ was the only active edge within this tie. Since $M$ is stable, $u'$ must have a partner $w'$ in $M$. Since $u'w'$ is a stable edge, it cannot have been rejected previously. Consequently, $w \prec_{u'} w'$. Thus, $u'w$ blocks $M$, which contradicts its stability.
        
        \item Line~\ref{li:Hall_reject}: $uw$ was rejected as an active edge incident to the critical set $U'$ in $G_A$.\\
        Let $W'=\mathcal{N}_{G_A}(U')$, $U'' = \{u \in U': M(u) \in W'\}$, and $W'' = \{w \in W': M(w) \in U'\}$. In words, $W'$ is the neighborhood of $U'$, while $U''$ and $W''$ represent the men and women in $U'$ and $W'$ who are paired up in~$M$. Due to our assumption, $u \in U''$ and $w \in W''$.
        
        We claim that  $|U' \setminus U''| < |U'|$ and $\df(U' \setminus U'') \geq \df(U')$, which contradicts the fact that $U'$ is critical. Since $U'' \neq \emptyset$, the first part holds. Note that $|U''|=|W''|$, so it suffices to show that $\mathcal{N}_{G_A}(U' \setminus U'') \subseteq W' \setminus W''$, because in that case
        \begin{IEEEeqnarray*}{rCl}
        	\df(U'\setminus U'') = |U'\setminus U''| - |\mathcal{N}_{G_A}(U'\setminus U'')| & \geq & |U' \setminus U''| - |W' \setminus W''| = \\
            & = & (|U'| - |W'|) - (|U''| - |W''|) = \\
            & = & |U'| - |W'| = \df(U'), 
        \end{IEEEeqnarray*}
         which would prove the second part of our claim.
         
         What remains to show is that $\mathcal{N}_{G_A}(U' \setminus U'') \subseteq W' \setminus W''$. Suppose that there exists an edge $ab$ in $G_A$ from $U'\setminus U''$ to~$W''$. We know that $b \in W''$, hence $a' = M(b) \in U''$ and, obviously, $a' \neq a \notin U''$. Moreover, $ab$ and $a'b$ are edges in $G_A$, thus both of them are active. Therefore, $a \sim_{b} a'$, for otherwise $b$ would have rejected one of them. In order to keep $M$ stable, $a$ must be paired up in $M$ with some woman $b'$. Since no stable edge has been rejected so far and $ab$ does not block $M$, therefore $b' \sim_{a} b$, thus $b'$ is in $a$'s proposal tie. Edge $ab'$ is stable and no stable edge has been rejected yet, thus $ab'$ is active along with~$ab$. Therefore, $ab' \in E(G_A)$ and $b' \in W'$. Moreover, $ab' \in M$, hence $a \in U''$ and $b' \in W''$, which contradicts the assumption that $a \notin U''$. \qedhere
    \end{itemize}
\end{proof}

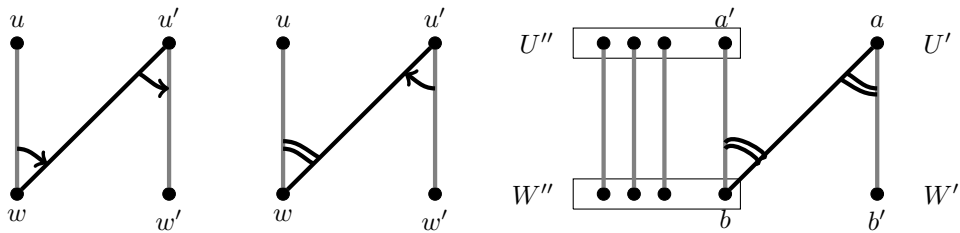
\begin{figure}[ht]
\centering
\begin{minipage}{0.25\textwidth}
\centering
\begin{tikzpicture}[scale=1, transform shape]

\node[vertex, label=above:$u$] (u) at (0, \d) {};
\node[vertex, label=above:$u'$] (u') at (\b, \d) {};
\node[vertex, label=below:$w$] (w) at (0, 0) {};
\node[vertex, label=below:$w'$] (w') at (\b, 0) {};

\draw [ultra thick, gray] (u) -- (w);
\draw [ultra thick, gray] (u') -- (w');
\draw [ultra thick] (u') -- (w);
\path [->,draw,ultra thick] ($ (w) !.3! (u) $) to [out=0,in=150] ($ (w) !.2! (u') $);
\path [<-,draw,ultra thick] ($ (u') !.3! (w') $) to [out=180,in=-30] ($ (u') !.2! (w) $);
\end{tikzpicture}
\end{minipage}\begin{minipage}{0.25\textwidth}
\centering
\begin{tikzpicture}[scale=1, transform shape]

\node[vertex, label=above:$u$] (u) at (0, \d) {};
\node[vertex, label=above:$u'$] (u') at (\b, \d) {};
\node[vertex, label=below:$w$] (w) at (0, 0) {};
\node[vertex, label=below:$w'$] (w') at (\b, 0) {};

\draw [ultra thick, gray] (u) -- (w);
\draw [ultra thick, gray] (u') -- (w');
\draw [ultra thick] (u') -- (w);
\draw [ultra thick] ($ (w) !.3! (u) $) to [out=0,in=150] ($ (w) !.2! (u') $);
\draw [ultra thick] ($ (w) !.35! (u) $) to [out=0,in=150] ($ (w) !.24! (u') $);
\path [->,draw,ultra thick] ($ (u') !.3! (w') $) to [out=180,in=-30] ($ (u') !.2! (w) $);
\end{tikzpicture}
\end{minipage}
\begin{minipage}{0.4\textwidth}
\centering
\begin{tikzpicture}[scale=1, transform shape]

\node[vertex, label=above:$a$] (a) at (0, \d) {};
\node[vertex, label=above:$a'$] (a') at (-\b, \d) {};
\node[vertex, label=below:$b'$] (b') at (0, 0) {};
\node[vertex, label=below:$b$] (b) at (-\b, 0) {};

\node[vertex] (a1) at (-1.4*\b, \d) {};
\node[vertex] (a2) at (-1.6*\b, \d) {};
\node[vertex] (a3) at (-1.8*\b, \d) {};
\node[vertex] (b1) at (-1.4*\b, 0) {};
\node[vertex] (b2) at (-1.6*\b, 0) {};
\node[vertex] (b3) at (-1.8*\b, 0) {};

\draw (-2*\b,0.9*\d) -- (-0.9*\b,0.9*\d) -- (-0.9*\b,1.1*\d) -- (-2*\b,1.1*\d) -- (-2*\b,0.9*\d);
\node[label=left:$U''$] (U'') at (-2*\b,\d) {};

\draw (-2*\b,-0.1*\d) -- (-0.9*\b,-0.1*\d) -- (-0.9*\b,0.1*\d) -- (-2*\b,0.1*\d) -- (-2*\b,-0.1*\d);
\node[label=left:$W''$] (W'') at (-2*\b,0) {};

\node[vertex, white, label=right:$U'$] (U') at (0.2*\b, \d) {};
\node[vertex, white, label=right:$W'$] (W') at (0.2*\b, 0) {};

\draw [ultra thick, gray] (a1) -- (b1);
\draw [ultra thick, gray] (a2) -- (b2);
\draw [ultra thick, gray] (a3) -- (b3);

\draw [ultra thick, gray] (a') -- (b);
\draw [ultra thick, gray] (a) -- (b');
\draw [ultra thick] (a) -- (b);
\path [draw,ultra thick] ($ (b) !.2! (a) $) to [out=0,in=60] ($ (b) !.3! (a') $);
\path [draw,ultra thick] ($ (b) !.24! (a) $) to [out=0,in=60] ($ (b) !.34! (a') $);
\path [draw,ultra thick] ($ (a) !.3! (b') $) to [out=180,in=-30] ($ (a) !.2! (b) $);
\path [draw,ultra thick] ($ (a) !.34! (b') $) to [out=180,in=-30] ($ (a) !.24! (b) $);
\end{tikzpicture}
\end{minipage}
\caption{The three cases in Claim~\ref{le:stable_edges}. Gray edges are in~$M$. The arrows point to the strictly preferred edges.}
\label{fi:stable_edges}
\end{figure}

\begin{claim}
\label{le:women_activated}
Women who have ever had an active edge must be matched in all stable matchings.
\end{claim}

\begin{proof}
Claim~\ref{le:stable_edges} shows that stable matchings allocate each man $u$ a partner not better than his final proposal tie. If a man $u$ proposed to woman $w$ and yet $w$ is unmatched in the stable matching $M$, then $uw$ blocks $M$, which contradicts the stability of~$M$.
\end{proof}

\begin{claim}
\label{le:cover_is_stable}
If our algorithm outputs a matching, then it is stable.
\end{claim}

\begin{proof}
We need to show that any maximum matching $M$ in $G_A$ is stable, if it covers all women who have ever held a proposal. Let $M$ be such a matching. Due to the exit criteria of the second phase (lines~\ref{li:Hall} and~\ref{li:u'emptyset}), $M$ covers all men. By contradiction, let us assume that $M$ is blocked by an edge $uw$. This can occur in three cases.
\begin{itemize}
	\item
    While $w$ is unmatched, $u$ does not prefer $M(u)$ to~$w$.\\
	Since $uw$ carried a proposal at the same time or before $uM(u) \in E(G_A)$ was activated, $w$ is a woman who has held an offer during the course of the algorithm. We assumed that all these women are matched in~$M$.
        
	\item While $w \prec_u M(u)$, $w$ does not prefer $M(w)$ to~$u$.\\
	The full tie at $u$ containing $uw$ must have been rejected in the algorithm, otherwise $uM(u)$ would not be an active edge. We know that either $u \prec_w M(w)$ or $u \sim_w M(w)$ holds. If $u \prec_w M(w)$, then $wM(w)$ had to be rejected when $u$ proposed to $w$, which contradicts our assumption that $wM(w) \in E(G_A)$. Hence, $u \sim_w M(w)$. Thus, when $uw$ and its full tie was rejected at $u$, $M(w)w$ also should have been rejected in a STRONG\_REJECT procedure, which leads to the same contradiction with $wM(w) \in E(G_A)$.
	
	\item While $u \prec_w M(w)$, $u$ does not prefer $M(u)$ to~$w$.\\
	Since $uM(u)$ is an active edge, $uw$ has carried a proposal, because $M(u)$ is not preferred to $w$ by~$u$. When $uw$ was proposed along, $w$ should have rejected $M(w)w$, to which $uw$ is strictly preferred. This contradicts our assumption that $wM(w) \in E(G_A)$.\qedhere
\end{itemize}
\end{proof}

\begin{claim}
	\label{le:stable_max_all_women}
    If $\I'$ admits a stable matching $M'$, then any maximum matching $M$ in the final $G_A$ covers all women who have ever held a proposal.
\end{claim}

\begin{proof}
	From Claims~\ref{le:dummy_women} and~\ref{le:women_activated} we know that $M'$ covers all women who have ever held a proposal and all men. It is also obvious that matching $M$ found in line~\ref{li:M} covers all men, for otherwise $U'$ could not have been the empty set in line~\ref{li:u'emptyset} and the execution would have returned to the first phase. This means that $|M|=|M'|$. On the other hand, all women covered by $M \subseteq E(G_A)$ are fit with active edges in $G_A$. Therefore, women covered by $M$ represent only a subset of women who have ever had an active edge, i.e.~the women covered by~$M'$. In order to $M$ and $M'$ have the same cardinality, they must cover exactly the same women. Thus, $M$ covers all women who have ever received a proposal.
\end{proof}

\begin{cor}
	\label{cor:admits_sm_outputs}
	If $\I$ admits a stable matching then our algorithm outputs one.
\end{cor}

\begin{proof}
	Since the edges between men and their dummy partners cannot be rejected, the algorithm will proceed to line~\ref{li:M}. Courtesy of Claim~\ref{le:stable_max_all_women}, the output $M$ covers all women who have ever received a proposal. According to Claim~\ref{le:cover_is_stable}, this matching is stable, and thus we output a stable matching of~$\I$.
\end{proof}

\section{Super-stability}
\label{se:super}

In super-stability, an edge outside of $M$ blocks $M$ if \emph{neither} of its end vertices \emph{prefer} their matching edge to it.

\begin{defi}[blocking edge for super-stability]
Edge $uw$ blocks $M$, if 
\begin{enumerate}
	\item $uw \notin M$;
	\item $ w \prec_u M(u)$ or $ w \sim_u M(u)$;
	\item $ u \prec_w M(w)$ or $ u \sim_w M(w)$.
\end{enumerate}
\end{defi}

The set of already investigated problems is most remarkable for super-stability, see Table~\ref{ta:all}. Up to posets on both sides, a polynomial algorithm is known to decide whether a stable solution exists~\cite{Irv94, Man13}. Even though it is not explicitly written there, a blocking edge in the super stable sense is identical to the definition of a blocking edge given in~\cite{FGK16}. It is shown there that if one vertex class has strictly ordered preference lists and the other vertex class has arbitrary relations, then determining whether a stable solution exists is $\NP$-complete, but if the second class has asymmetric lists, then the problem becomes tractable.

We first show that a polynomial algorithm exists up to partially ordered relations on one side and asymmetric relations on the other side. Our algorithm can be seen as an extension of the one in~\cite{FGK16}. Our added contributions are a more sophisticated proposal routine and the condition on stability in the output. These are necessary as men are allowed to have acyclic preferences instead of strictly ordered lists, as in~\cite{FGK16}. Finally, we prove that acyclic relations on both sides make the problem hard.

\begin{restatable}{theorem}{thsuppol}
\label{th:sup_pol}
If one side has posets as preferences, while the other side has asymmetric pairwise preferences, then deciding whether the instance admits a super stable matching can be done in $\Ord{n^2m}$ time. 
\end{restatable}

We prove this theorem by designing an algorithm that produces a stable matching or a proof for its nonexistence, see 
Algorithm~\ref{alg:super_stable_poset_assymetric}. We assume men to have posets as preferences and women to have asymmetric relations. We remark that non-empty posets always have a non-empty set of \emph{maximal elements}: these are the ones that are not dominated by any other element. Women in the set of maximal elements are called \emph{maximal} women.

At start, an arbitrary man proposes to one of his maximal women. An offer from $u$ is temporarily accepted by $w$ if and only if $u \prec_w u'$ for every man $u' \neq u$ who has ever proposed to~$w$. This rule forces each woman to reproof her current match every time a new proposal arrives. Accepted offers are called \emph{engagements}. The proposal edges or engagements not meeting the above requirement are immediately deleted from the graph. Each man then reexamines the poset of women still on his list. If any of the maximal women is not holding an offer from him, then he proposes to her. The process terminates and the output is generated when no man has maximal women he has not proposed to. Notice that while women hold at most one proposal at a time, men might have several engagements in the output.

\begin{algorithm}
    \caption{Super stable matching with posets and asymmetric relations}
    \textbf{Input}: $\I=(U, W, E, \mathcal{R}_{U}, \mathcal{R}_{W})$; $\mathcal{R}_{U}$: posets, $\mathcal{R}_{W}$: asymmetric.
    
    \begin{algorithmic}[1]
        \setcounter{ALG@line}{\value{myalgcounter}}
        
        \While{there is a man $u$ who has not proposed to a maximal woman $w$} \label{li:super:man_to_propose}
            \State $u$ proposes to $w$ \label{li:super:man_to_propose_2}
            \If{$u \prec_w u'$ for all $u' \in U$ who has ever proposed to $w$} \label{li:super:previous_proposes}
                \State $w$ accepts the proposal of $u$, $uw$ becomes an engagement \label{li:super:accept}
            \Else
                \State $w$ rejects the proposal and deletes $uw$ \label{li:super:reject}
            \EndIf
            \If{$w$ had a previous engagement to $u'$ and $u \prec_w u'$ or $u \sim_w u'$} \label{li:super:previous_engagement}
                \State $w$ breaks the engagement to $u'$ and deletes $u'w$ \label{li:super:previous_engagement_break}
            \EndIf
        \EndWhile \medskip
        
        \State let $M$ be the set of engagements \label{li:super:engagement_graph}
        \label{li:super:}
        \If{$M$ is a matching that covers all women who have ever received a proposal} \label{li:super:engagement_graph_cond}
            \State STOP, OUTPUT $M$ and ``$M$ is a super stable matching.'' \label{li:super:solution}
        \Else
            \State STOP, OUTPUT ``There is no super stable matching.'' \label{li:super:no_solution}
        \EndIf
        
    \end{algorithmic}
    
    \label{alg:super_stable_poset_assymetric}
\end{algorithm}

The correctness and time complexity of our algorithm is shown in the Appendix, where we prove that the set of engagements $M$ is a matching that covers all women who ever received a proposal if and only if the instance admits a stable matching. 

\begin{restatable}{theorem}{thststnp}
\label{th:st_st_np}
If both sides have acyclic pairwise preferences, then determining whether a super stable matching exists is $\NP$-complete, even if each agent finds at most four other agents acceptable.
\end{restatable}
\section{Conclusion and open questions}
We completed the complexity study of the stable marriage problem with pairwise preferences. Despite of the integrity of this work, our approach opens the way to new research problems. 

The six degrees of orderedness can be interpreted in the non-bipartite stable roommates problem as well. For strictly ordered preferences, all three notions of stability reduce to the classical stable roommates problem, which can be solved in $\mathcal{O}(m)$ time~\cite{Irv85}. The weakly stable variant becomes $\NP$-complete already if ties are present~\cite{Ron86}, while the strongly stable version can be solved with ties in polynomial time, but it is $\NP$-complete for posets. The complexity analysis of these cases is thus complete. Not so for super-stability, for which there is an $\mathcal{O}(m)$ time algorithm for preferences ordered as posets~\cite{IM02}, while the case with asymmetric preferences was shown here to be $\NP$-complete for bipartite instances as well. We conjecture that the intermediate case of acyclic preferences is also polynomially solvable and the algorithm of Irving and Manlove can be extended to it. 

The  Rural Hospitals Theorem~\cite{GS85} states that the set of matched vertices is identical in all stable matchings. It has been shown to hold for strongly and super stable matchings~\cite{IMS00,Man02} and fail for weak stability, if preferences contain ties---even for non-bipartite instances. We remark that these results carry over even to the most general pairwise preference setting. To see this, one only needs to sketch the usual alternating path argument: assume that there is a vertex $v$ that is covered by a stable matching $M_1$, but left uncovered by another stable matching~$M_2$. Then, $M_1(v)$ must strictly prefer its partner in $M_2$ to $v$, otherwise edge $vM_1(v)$ blocks~$M_2$. Iterating this argument, we derive that such a $v$ cannot exist. The Rural Hospitals Theorem might indicate a rich underlying structure of the set of stable matchings. Such results were shown in the case of preferences with ties. Strongly stable matchings are known to form a distributive lattice~\cite{Man02}, and there is a partial order with $\mathcal{O}(m)$ elements representing all strongly stable matchings~\cite{KPG16}. However, once posets are allowed in the preferences, the lattice structure falls apart~\cite{Man02}. The set of super stable matchings has been shown to form a distributive lattice if preferences are expressed in the form of posets~\cite{Man02,Spi95}. The questions arise naturally: does this distributive lattice structure carries over to more advanced preference structures in the super stable case? Also, even if no distributive lattice exists on the set of strongly stable matchings, is there any other structure and if so, how far does it extend in terms of orderedness of preferences?

\bibliography{mybib}
\newpage
\section*{\hypertarget{app:appendix}{Appendix}}
\subsection*{Weak stability}
\thweaknp*

\begin{proof}
The $\NP$-complete problem we reduce to our problem is \textsc{(2,2)-e3-sat}~\cite{BKS03}. Its input is a Boolean formula $B$ in CNF, in which each clause comprises exactly 3 literals and each variable appears exactly twice in unnegated and exactly twice in negated form. The decision question is whether there exists a truth assignment satisfying~$B$.

When constructing graph $G$ to a given Boolean formula $B$, we keep track of the three literals in each clause and the two unnegated and two negated appearances of each variable. Each appearance is represented by an interconnecting edge, running between the corresponding variable and clause gadgets. The graphs underlying our gadgets resemble gadgets in earlier hardness proofs~\cite{BMM10}, but the preferences are designed specifically for our problem. Figure~\ref{fi:np_weakly} illustrates our construction, in particular, the preference relations in it.

\begin{figure}[htbp]

\begin{minipage}[h]{0.4\textwidth}
\vspace{3mm}
$$
\begin{array}{ll}
   	t: & \text{strict list: } x \prec \bar{x}\\
    f: & \text{strict list: } \bar{x} \prec x\\
    \hline
	x: & f \prec t, t \prec u, u \prec f\\
	\bar{x}: & t \prec f, t \prec u, u \prec f
\end{array}
$$
\end{minipage}
\begin{minipage}[h]{0.4\textwidth}
$$\begin{array}{ll}
	u_1: & \text{strict list: } w_3 \prec w_2 \prec \tilde{x}_{u_1} \prec w_1\\
	u_2: & \text{strict list: } w_3 \prec w_2 \prec \tilde{x}_{u_2} \prec w_1\\
	u_3: & \text{strict list: } w_3 \prec w_2 \prec \tilde{x}_{u_3} \prec w_1\\
    \hline
	w_1: & \emptyset\\
	w_2: & \emptyset\\
	w_3: & \emptyset
\end{array}$$
\end{minipage}


\begin{tikzpicture}[scale=0.8, transform shape]

\tikzstyle{vertex} = [circle, draw=black, scale=0.7]
\tikzstyle{specvertex} = [circle, draw=black, scale=0.7]
\tikzstyle{edgelabel} = [circle, fill=white, scale=0.5, font=\huge]
\pgfmathsetmacro{\d}{1.05}
\pgfmathsetmacro{\b}{2}

\node[vertex, label={[label distance=9mm]below:$x$}] (x) at (0, \d*3) {};
\node[vertex, label={[label distance=9mm]below:$\bar{x}$}] (ox) at (0, 0) {};
\node[vertex, label=left:$t$] (t) at (-\b*2, \d*3) {};
\node[vertex, label=left:$f$] (f) at (-\b*2, 0) {};

\draw [ultra thick] (x) -- node[edgelabel, near end] {1} (t);
\draw [ultra thick] (ox) -- node[edgelabel, near end] {1} (f);
\draw [ultra thick] (ox) -- node[edgelabel, near end] {2} (t);
\draw [ultra thick] (x) -- node[edgelabel, near end] {2} (f);

\foreach \from in {x, ox}{
   \draw [ultra thick, dotted] (\from) to[out=60,in=180, distance=0.4cm ] ($ (\from) + (4/3*\d,1/4*\b) $);
   \draw [ultra thick, dotted] (\from) to[out=-60,in=180, distance=0.4cm ] ($ (\from) + (4/3*\d,-1/4*\b) $);
    }
    
\path [->,draw,ultra thick] ($ (x) !.20! (t) $) to [out=270,in=135] ($ (x) !.15! (f) $);
\path [->,draw,ultra thick] ($ (ox) !.20! (f) $) to [out=90,in=-135] ($ (ox) !.15! (t) $);
\path [->,draw,ultra thick] ( $(x) !.20! ($ (x) + (4/3*\d,-2/4*\b) $)$) to [out=45,in=45, distance=5mm] ($ (x) !.1! (t) $);
\path [->,draw,ultra thick] ($ (x) !.25! ($ (x) + (4/3*\d,2/4*\b) $)$) to [out=135,in=45] ($ (x) !.20! (t) $);
\path [->,draw,ultra thick] ($ (x) !.15! (f) $) to [out=-45,in=-135] ($ (x) !.25! ($ (x) + (4/3*\d,-2/4*\b) $)$);
\path [->,draw,ultra thick] ($ (x) !.2! (f) $) to [out=-45,in=-45, distance=10mm] ($ (x) !.35! ($ (x) + (4/3*\d,2/4*\b) $)$);

\path [->,draw,ultra thick] ($ (ox) !.20! ($ (ox) + (4/3*\d,-2/4*\b) $)$) to [out=90,in=0] ($ (ox) !.1! (t) $);
\path [->,draw,ultra thick] ($ (ox) !.20! ($ (ox) + (4/3*\d,3/4*\b) $)$) to [out=90,in=45] ($ (ox) !.15! (t) $);
\path [->,draw,ultra thick] ($ (ox) !.1! (f) $) to [out=-90,in=-145, distance=4mm] ($ (ox) !.3! ($ (ox) + (4/3*\d,-2/4*\b) $)$);
\path [->,draw,ultra thick] ($ (ox) !.17! (f) $) to [out=-90,in=-45, distance=13mm] ($ (ox) !.35! ($ (ox) + (4/3*\d,2/4*\b) $)$);

\node[fill=white] (u1) at (0, -\d*1.6) {};

\end{tikzpicture}\hspace{8mm}
%
%
\begin{tikzpicture}[scale=0.8, transform shape]

\tikzstyle{vertex} = [circle, draw=black, scale=0.7]
\tikzstyle{specvertex} = [circle, draw=black, scale=0.7]
\tikzstyle{edgelabel} = [circle, fill=white, scale=0.5, font=\huge]
\pgfmathsetmacro{\d}{1.05}
\pgfmathsetmacro{\b}{2}

\node[vertex, label=left:$u_1$] (u1) at (0, -\d*3) {};
\node[vertex, label=left:$u_2$] (u2) at (0, -\d*6) {};
\node[vertex, label=left:$u_3$] (u3) at (0, -\d*9) {};

\node[vertex, label=right:$w_1$] (w1) at (\b*3, -\d*3) {};
\node[vertex, label=right:$w_2$] (w2) at (\b*3, -\d*6) {};
\node[vertex, label=right:$w_3$] (w3) at (\b*3, -\d*9) {};

\draw [ultra thick] (u1) -- node[edgelabel, near start] {4} (w1);
\draw [ultra thick] (u1) -- node[edgelabel, near start] {2} (w2);
\draw [ultra thick] (u1) -- node[edgelabel, near start] {1} (w3);

\draw [ultra thick] (u2) -- node[edgelabel, near start] {4} (w1);
\draw [ultra thick] (u2) -- node[edgelabel, near start] {2} (w2);
\draw [ultra thick] (u2) -- node[edgelabel, near start] {1} (w3);

\draw [ultra thick] (u3) -- node[edgelabel, near start] {4} (w1);
\draw [ultra thick] (u3) -- node[edgelabel, near start] {2} (w2);
\draw [ultra thick] (u3) -- node[edgelabel, near start] {1} (w3);

\path [double,draw,ultra thick] ($ (w1) !.20! (u1) $) to [out=-90,in=135] ($ (w1) !.13! (u3) $);
\path [double,draw,ultra thick] ($ (w2) !.15! (u1) $) to [out=230,in=130] ($ (w2) !.15! (u3) $);
\path [double,draw,ultra thick] ($ (w3) !.13! (u1) $) to [out=-135,in=90] ($ (w3) !.20! (u3) $);

\foreach \from in {u1, u2, u3}
   \draw [ultra thick, dotted] (\from) to[out=130,in=30, distance=0.8cm ] node[edgelabel, near start] {3} ($ (\from) + (-4/3*\d,0) $);
\end{tikzpicture}


\caption{A variable gadget to the left and a clause gadget to the right. Strict lists are to be found at $t$, $f$, and $u$-vertices, while the rest of the vertices have asymmetric relations. Interconnecting edges are dotted. The arrows point to the preferred edge, while double lines denote incomparability.}
\label{fi:np_weakly}
\end{figure}
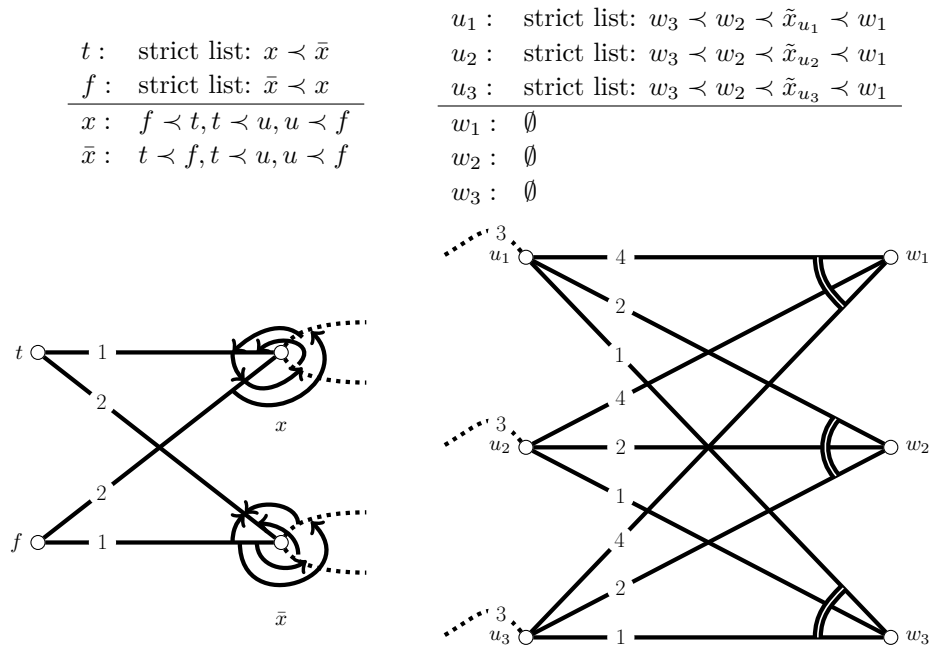

The variable gadget comprises a 4-cycle $t,\bar{x},f,x$ and four interconnecting edges, two of which are incident to $x$, and the remaining to are adjacent to~$\bar{x}$. These four edges are connected to $u$-vertices in clause gadgets. In each variable gadget, $x$  symbolizes the unnegated occurrences of the variable, while $\bar{x}$ stands for the negated occurrences.

The clause gadget consists of a complete bipartite graph on six vertices, where one side is equipped with interconnecting edges. This side represents the three literals in the clause. Each interconnecting edge runs to vertex $x$ or $\bar{x}$ in the variable gadget of the occurring unnegated or negated variable~$x$.

\begin{claim}
	If there is a weakly stable matching $M$ in $G$, then there is a truth assignment to~$B$.
\end{claim}

First we show that $t$ and $f$ must be matched in all stable matchings. If $t$ is unmatched, then both $x$ and $\bar{x}$ must be matched to a vertex to which $t$ is not preferred. The only such vertex is $f$, which leads to a contradiction with the matching property of~$M$. If $f$ is unmatched, then neither $x$ nor $\bar{x}$ is allowed to be matched to $t$, which we just showed to be impossible. Thus, any stable matching contains either $\left\{tx,f\bar{x} \right\}$ or $\left\{fx,t\bar{x} \right\}$ for each variable gadget. We set a variable to be true if $\left\{tx,f\bar{x} \right\} \in M$ and to false if $\left\{fx,t\bar{x} \right\} \in M$.

Another consequence of $M$ covering all $t$ and $f$ vertices, is that $M$ contains no interconnecting edge. From this follows that $M$ restricted to an arbitrary clause gadget must be a perfect matching.

The preferences in the clause gadgets are set so that out of the three interconnecting edges running to a clause gadget, exactly one dominates $M$ at the clause gadget, namely the edge incident to vertex $u_i$ paired up with~$w_1$. We know that $M$ is stable, therefore, this dominating interconnecting edge must be dominated by its other end vertex. This is only possible if the variable is set to true if the literal was unnegated, and to false if the literal was in negated form. Thus, we have found a satisfied literal in each clause.
\end{proof}

\begin{claim}
	If there is a truth assignment to $B$, then there is a stable matching $M$ in~$G$.
\end{claim}

In each variable gadget belonging to a true variable, $\left\{tx,f\bar{x} \right\}$ is chosen, whereas all gadgets corresponding to a false variable contribute the edges $\left\{fx,t\bar{x} \right\}$. In each clause, there is at least one true literal. We match the vertex representing the appearance of this literal to $w_1$ and match $w_2$ and $w_3$ arbitrarily.

No edge inside of a gadget blocks $M$, because it is a perfect matching inside each gadget and the preferences are either cyclic (variable gadget), or one side is indifferent (clause gadget). An interconnecting edge dominates $M$ at the clause gadget if and only if it corresponds to the chosen literal satisfying the clause. Our rules set exactly this literal to be satisfied in the variable gadget, i.e.~this literal is paired up with $t$, which is strictly preferred to the corresponding interconnecting edge.


\subsection*{Strong stability}

\medskip
{\color{black}
\textbf{Analysis and time complexity of Algorithms~\ref{alg:strong_SMTAs} and~\ref{alg:reject}.} \label{analysis:alg_strong}
We suppose that $G$ is represented by adjacency lists belonging to $|U| + |W| = n$ vertices and that there are $|E| = m$ acceptable edges. Since zero-degree vertices do no interfere with the existence or content of stable matchings, it may be assumed that each vertex has at least one edge, which results in $max\{|U|,|W|\} \leq m$, hence $n = |U| + |W| \leq 2m$ and $n=\mathcal{O}(m)$. Relations in $\mathcal{R}_U$ are lists with ties, hence they can be incorporated into the adjacency lists by using a delimiter symbol between ties. However, relations in $\mathcal{R}_W$ are to be represented as general relations with at most $\binom{|U|}{2} = \mathcal{O}(n^2)$ elements. The cost of the execution of the algorithm on an instance $\I$ is estimated by the number of accesses to the data structures representing neighbors of vertices and the relations between them.

Firstly, a lower bound of the size of input is provided by the size of the graph, as usual. Note that relations in $\mathcal{R}_W$ may be empty sets, so this is a sharp lower bound. Hence, the input size is $\Omega(n+m)$.

Secondly, non-trivial operations are to be committed on a data structure holding asymmetric relations. Our algorithm uses the following operation primitives: finding all men $u'$ such that $u \prec u'$ with respect to $R_w$ and rejecting $u'w$, finding edges incomparable to $uw$ with respect to $R_w$ and rejecting them. These primitives can take up as many as $\mathcal{O}(n^2)$ steps. Let us denote the maximum cost of any such primitive by $\xi$.

In order to decrease running time, all information regarding edges are to be maintained. More specifically, the state of an edge as being inactive, active or rejected is stored. Moreover, for every $u \in U$, we store the fact whether $u$ has been a vertex because of which in Algorithm~\ref{alg:reject} edges of type $u'w$ are rejected where $u' \sim_w u$. Reasonable work is spared if $u$ plays the same role again later.

Now, adding dummy women to the list of men is done in $\Ord{n}$ time in total. Besides, each edge is proposed along at most once and proposals are to be done in order of the adjacency list of men, so the total cost of proposals is $\Ord{m}$. Furthermore, beware that for a given edge $uw$, rejecting edges $u'w$ to whom $uw$ is strictly preferred, and rejecting incomparable edges $u'w$ are done at most once, each of them contributing a cost of $\xi$. The graph $G_A$ need not be constructed separately, since active edges are marked due to our previous considerations. Subsequently, apart from finding maximal matchings and critical sets in $G_A$, the cost of our algorithm is bounded by $\Ord{n + m + 2m\xi} = \Ord{m\xi}$.
 
As far as maximum matchings and critical sets are concerned, the well-founded technique described by Irving~\cite{Irv94} is reapplied here. As already stated previously, the critical set is calculated from a maximum matching by taking the uncovered men and all men reachable from the uncovered men via an alternating path. The standard algorithm for determining maximum matchings launches parallel BFS-algorithms from uncovered men to find augmenting paths. An interesting property of the execution is that whenever it finishes---because no alternating path was augmenting,---the critical set is computed as well. Therefore critical sets are automatically yielded with the use of the Hungarian method, for which one only needs to store the occurring vertices.

Although we could apply the Hungarian method in each execution of the second phase, we wish to reduce the cost of execution by storing information from previous iterations. Note that the Hungarian method commences from an arbitrary matching and augments that one. Let the augmentation start from the remnants of the maximum matching found in the previous iteration. Let $M_i$, $C_i$, $x_i, (i \geq 1)$ denote the maximum matching found in the $i^\text{th}$ iteration of the second phase, the critical set with respect to $M_i$, and the number of edges rejected between the $i^\text{th}$ and $(i+1)^\text{th}$ execution of the Hungarian method, respecticely. In the first iteration the augmenting path algorithm is executed from scratch taking $\Ord{|U|m} = \Ord{nm}$ time. After the $i^\text{th}$ iteration we reject $x_i$ edges. Since each man in $C_i$ had at least one edge in $G_A$, at least $(|U|-|C_i|)-(x_i-|C_i|)=|U|-x_i$ men are still paired to women via active edges, if that number is positive. In that case, the $(i+1)^\text{th}$ iteration starts BFS-algorithms from $x_i$ vertices. Let  $L$ be the total number of iterations, in $k$ of which $x_i \geq |U|$, i.e.~the augmenting path algorithm is run from scratch. The time complexity, therefore, is $\Ord{nm + kmn + m\sum_{L-k \text{ iter}} x_i}$, where the summation is done for the rest of $x_i$'s corresponding to the remaining $L-k$ iterations. The time complexity, in the other $k$ iterations $n \leq x_i$, therefore $kn + \sum_{L-k \text{ iter}} x_i \leq \sum_{i=1}^{L} x_i \leq m$, because not more than $m$ edges may be rejected and no edge is rejected more than once. Hence the running time related to maximum matchings and critical sets is $\Ord{nm + m \cdot (kn + \sum_{L-k \text{ iter}} x_i)} = \Ord{nm + m \cdot m} = \Ord{m^2}$.

In conclusion, the total time complexity of the algorithm is $\Ord{m\xi + m^2} = \Ord{mn^2 + m^2}$, while the size of the input is $\Omega(n+m)$.


\subsection*{Super-stability}

\begin{theorem}
The output of Algorithm~\ref{alg:super_stable_poset_assymetric} is a matching that covers all women who ever received a proposal if and only if the instance admits a stable matching. 
\end{theorem}

\begin{claim}
\label{cl:super:output_is_stable}
	If the output of the algorithm is a matching that covers all women who ever received a proposal, then it is stable.
\end{claim}

\begin{proof}
	Assume that an edge $uw$ blocks the output matching~$M$. We investigate two cases. \begin{itemize}
	    \item Man $u$ has proposed to $w$.\\
	    We know that $w$ got engaged to a man $M(w)$, for whom $M(w) \prec_w u$ holds. This contradicts our assumption on $uw$ being a blocking edge.
	    \item Man $u$ has not proposed to~$w$.\\
	   There must be an edge $uw'$ not deleted so that $w' \prec_u w$. For $uw$ blocks $M$, $w' \neq M(u)$, thus $uw'$ has not been proposed along. Therefore, there is another edge $uw''$ not yet deleted so that $w'' \prec_u w' \prec_u M(u)$. Due to the transitivity of relations on the men's side and the finiteness of the vertex set, the iteration of this argument leads to a contradiction.\qedhere 
	\end{itemize}
\end{proof}
    
    \noindent The opposite direction we prove in Claims~\ref{cl:deleted_edge} to~\ref{cl:sust}.
    \begin{claim}
    \label{cl:deleted_edge}
    	If an edge was deleted in the algorithm, then no stable matching contains it.
    \end{claim}
    \begin{proof}
    	Let $uw$ be the first edge deleted by the algorithm even though it is part of a stable matching~$S$. The reason of the deletion was that $w$ received an offer from $u'$ for which $u' \prec_w u$ or $u' \sim_w u$. Since $u'w \notin S$ does not block $S$, $u'$ is matched in $S$ and $S(u') \prec_{u'} w$. Due to the monotonicity of proposals, $u'$ had proposed to $S(u')$ before proposing to $w$, but $u'S(u')$ was deleted. This contradicts our assumption on $uw$ being the first deleted stable edge.
    \end{proof}
    
    \begin{claim}
    \label{cl:unmatched_woman}
    	If a woman $w$ has ever received a proposal in our algorithm, then $w$ must be matched in all stable matchings.
    \end{claim}
    \begin{proof}
    	Assume that $uw$ carried a proposal at some point, yet $w$ is unmatched in a stable matching~$S$. In order to stop $uw$ from blocking $S$, $u$ is matched in $S$ and $S(u) \prec_u w$. This implies that $uS(u)$ was deleted before the proposal along $uw$ could be sent, which contradicts Claim~\ref{cl:deleted_edge}.
    \end{proof}
    
    \begin{claim}
\label{cl:sucover}
	If there is a stable matching $S$, then the set of engagements $M$ computed in line~\ref{li:super:engagement_graph} covers all women who have ever received a proposal.
\end{claim} 
\begin{proof}
	Assume that woman $w$ has received a proposal, but she is not covered in~$M$. Claim~\ref{cl:unmatched_woman} shows that $w$ is matched in $S$, while Claim~\ref{cl:deleted_edge} proves that $uw \in S$ was not proposed along. The latter implies that $u$ has at least one engagement edge in~$M$. 
	For the same reason, $w$ is not preferred to $M(u)$ by $u$ for all $uM(u) \in M$. To stop $uM(u)$ from blocking $S$, $M(u)$ must have a partner in $S$ who is preferred to~$u$. This edge obviously never carried a proposal, otherwise $uM(u)$ could not be in~$M$. We iterate this argument until the cycle closes. This cannot happen \begin{inparaenum}[1)]
	    \item at an $S$-edge running to an already visited vertex in $U$, because $S$ is a matching;
	    \item at an $M$-edge running to an already visited vertex in $W\setminus w$, because women keep at most one proposal edge;
	    \item at $w$, because $w$ is unmatched in~$M$.
	\end{inparaenum} In all cases, we arrived to a contradiction.\qedhere
\end{proof}

\begin{claim}
\label{cl:sust}
If there is a stable matching $S$, then the set of engagements $M$ computed in line~\ref{li:super:engagement_graph} is a matching.
\end{claim} 
\begin{proof}
As already mentioned, the only reason for $M$ not being a matching is that a man $u$ has more than one edges in~$M$. Since $S$ is a matching, not all of these are in~$S$. Let us denote an arbitrary edge of $u$ in $M \setminus S$ by~$uw$. $uw$ is an engagement and no stable edges are deleted, therefore $M(u)$ (either a woman or $\emptyset$) is not preferred to $w$. Thus, from the stability of~$S$, $w$ must have a strictly preferred edge in~$S$. Moreover, we also know that $u_1 = S(w)$ has never proposed to~$w$, otherwise $uw$ could not be in~$M$. So there exists a maximal woman $w_1 \in M(u_1)$ such that, $w_1 \prec_{u_1} w$. 

Due to analogous arguments, this preference path must continue. Since the graph has a finite number of edges, it must return to a vertex already visited. This recurring vertex cannot be in $U\setminus u$, because no vertex in $U$ has more than one edge in $S$ and similarly, it cannot be a vertex in $W$, because no woman has more than one edge in $M$. The only option therefore is that the cycle closes at~$u$. In this case, $uS(u) \notin M$, thus $M$ must have another edge in $M \setminus S$, because there are at least two edges in $M$ incident to~$u$. Repeating the same deductions, we arrive to another augmenting path that ends in a cycle at $u$ via another edge from~$S$. This contradicts the fact that $S$ is a matching.\qedhere
\end{proof} 
 

\textbf{Analysis and time complexity of Algorithm~\ref{alg:super_stable_poset_assymetric}}
We use a similar data structure to the one applied in the analysis of Algorithms~\ref{alg:strong_SMTAs} and~\ref{alg:reject}. The difference emerges from the poset preference structure on one side. We store the entire partial order for each man, 
given as a Hasse-diagram of the underlying directed acyclic graph of the poset, equipped with a dummy woman, from whom there is a directed edge to all initially maximal women. The cost of the execution is again grasped by the number of accesses to these data structures.

Since relations can be empty as well, the size of the input is analogously lower bounded by $\Om{n+m}$. The assumption of Hasse-diagrams allows a straightforward check whether all maximal women have been proposed to. The initial maximal set is the women directly connected to the dummy woman. Each time a woman $w$ turns down a proposal, the candidates of being promoted to maximal state are the women directly connected to $w$ in the Hasse-diagram. Therefore the cost of submitting proposals does not exceed $\Ord{m}$. The rest of the while loop, from lines~\ref{li:super:previous_proposes} to~\ref{li:super:previous_engagement_break}, concerns the asymmetric relations on the woman's side. One needs to iterate through the relations belonging to the woman in question and check whether the new proposal is strictly preferred to all previous proposals, and whether the previous engager is strictly preferred to the new one. This operation primitive has cost $\xi = \Ord{n^2}$. It is also remarked that, although we ``delete'' rejected proposal edges, in reality they could simply be marked as rejected. Then, checking previous proposals is meaningful again. Last but not least, the computation of $M$ and the examination of the output condition can be done in $\Ord{m}$ time, because engagements are marked anyway. Consequently, the time complexity of the algorithm is $\Ord{m \cdot \xi} + \Ord{m} = \Ord{n^2m}$. 

\thststnp*

\begin{proof}
The $\NP$-complete problem we reduce to our problem is again \textsc{(2,2)-e3-sat}~\cite{BKS03}. Our construction follows the same logic as the one in the proof of Theorem~\ref{th:weak_np}, however, the preferences are set differently, see Figure~\ref{fi:np_super}.

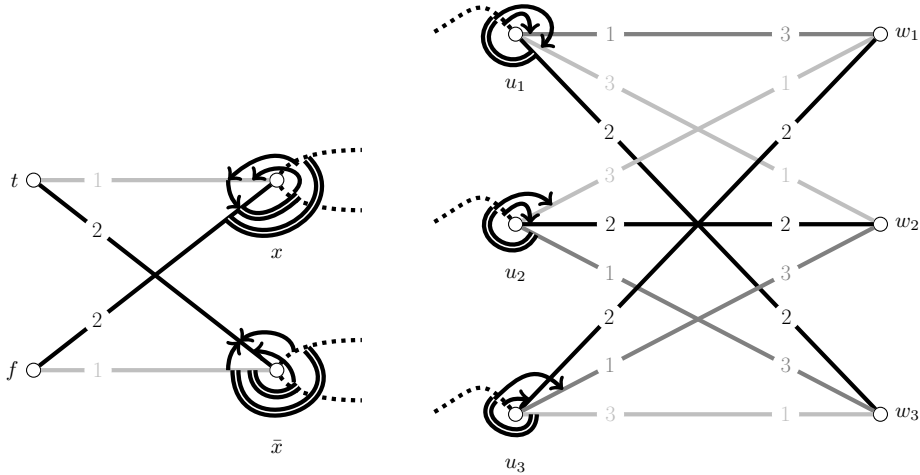
\begin{figure}[t]
\begin{minipage}[h]{0.3\textwidth}
$$\begin{array}{ll}
   	t: & \text{strict list: }  x \prec \bar{x}\\
    f: & \text{strict list: }  \bar{x} \prec x\\
    \hline
	x: & f \prec t, t \prec u, u \sim f\\
	\bar{x}: & t \prec f, t \prec u, u \sim f
\end{array}
$$
\end{minipage}
\begin{minipage}[h]{0.55\textwidth}
$$\begin{array}{ll}
	u_1: & w_1\prec x, w_2\prec x, x\sim w_3; \text{ strict list: } w_1\prec w_3\prec w_2\\
	u_2: & w_1\prec x, w_2\prec x, x\sim w_3; \text{ strict list: } w_3\prec w_2\prec w_1\\
	u_3: & w_1\prec x, w_2\prec x, x\sim w_3; \text{ strict list: } w_2\prec w_1\prec w_3\\
    \hline
	w_1: & \text{strict list: } u_2 \prec u_3 \prec u_1\\
	w_2: & \text{strict list: } u_1 \prec u_2 \prec u_3\\
	w_3: & \text{strict list: } u_3 \prec u_1 \prec u_2
\end{array}$$
\end{minipage}

\begin{tikzpicture}[scale=0.8, transform shape]

\tikzstyle{vertex} = [circle, draw=black, scale=0.7]
\tikzstyle{specvertex} = [circle, draw=black, scale=0.7]
\tikzstyle{edgelabel} = [circle, fill=white, scale=0.5, font=\huge]
\pgfmathsetmacro{\d}{1.05}
\pgfmathsetmacro{\b}{2}

\node[vertex, label={[label distance=9mm]below:$x$}] (x) at (0, \d*3) {};
\node[vertex, label={[label distance=9mm]below:$\bar{x}$}] (ox) at (0, 0) {};
\node[vertex, label=left:$t$] (t) at (-\b*2, \d*3) {};
\node[vertex, label=left:$f$] (f) at (-\b*2, 0) {};

\draw [ultra thick, gray!50] (x) -- node[edgelabel, near end] {1} (t);
\draw [ultra thick, gray!50] (ox) -- node[edgelabel, near end] {1} (f);
\draw [ultra thick] (ox) -- node[edgelabel, near end] {2} (t);
\draw [ultra thick] (x) -- node[edgelabel, near end] {2} (f);;

\foreach \from in {x, ox}{
   \draw [ultra thick, dotted] (\from) to[out=60,in=180, distance=0.4cm ] ($ (\from) + (4/3*\d,1/4*\b) $);
   \draw [ultra thick, dotted] (\from) to[out=-60,in=180, distance=0.4cm ] ($ (\from) + (4/3*\d,-1/4*\b) $);
    }
    
\path [->,draw,ultra thick] ($ (x) !.20! (t) $) to [out=270,in=135] ($ (x) !.15! (f) $);
\path [->,draw,ultra thick] ($ (ox) !.20! (f) $) to [out=90,in=-135] ($ (ox) !.15! (t) $);
\path [->,draw,ultra thick] ( $(x) !.20! ($ (x) + (4/3*\d,-2/4*\b) $)$) to [out=45,in=45, distance=5mm] ($ (x) !.1! (t) $);
\path [->,draw,ultra thick] ($ (x) !.25! ($ (x) + (4/3*\d,2/4*\b) $)$) to [out=135,in=45] ($ (x) !.20! (t) $);
\path [double,draw,ultra thick] ($ (x) !.15! (f) $) to [out=-45,in=-135] ($ (x) !.25! ($ (x) + (4/3*\d,-2/4*\b) $)$);
\path [double,draw,ultra thick] ($ (x) !.2! (f) $) to [out=-45,in=-45, distance=10mm] ($ (x) !.35! ($ (x) + (4/3*\d,2/4*\b) $)$);

\path [->,draw,ultra thick] ($ (ox) !.20! ($ (ox) + (4/3*\d,-2/4*\b) $)$) to [out=90,in=0] ($ (ox) !.1! (t) $);
\path [->,draw,ultra thick] ($ (ox) !.20! ($ (ox) + (4/3*\d,3/4*\b) $)$) to [out=90,in=45] ($ (ox) !.15! (t) $);
\path [double,draw,ultra thick] ($ (ox) !.1! (f) $) to [out=-90,in=-135, distance=4mm] ($ (ox) !.25! ($ (ox) + (4/3*\d,-2/4*\b) $)$);
\path [double,draw,ultra thick] ($ (ox) !.17! (f) $) to [out=-90,in=-45, distance=13mm] ($ (ox) !.35! ($ (ox) + (4/3*\d,2/4*\b) $)$);

\node[fill=white] (u1) at (0, -\d*1.6) {};

\end{tikzpicture}\hspace{8mm}\begin{tikzpicture}[scale=0.8, transform shape]

\tikzstyle{vertex} = [circle, draw=black, scale=0.7]
\tikzstyle{specvertex} = [circle, draw=black, scale=0.7]
\tikzstyle{edgelabel} = [circle, fill=white, scale=0.5, font=\huge]
\pgfmathsetmacro{\d}{1.05}
\pgfmathsetmacro{\b}{2}
\node[vertex, label={[label distance=5mm]below:$u_1$}] (u1) at (0, -\d*3) {};
\node[vertex, label={[label distance=5mm]below:$u_2$}] (u2) at (0, -\d*6) {};
\node[vertex, label={[label distance=5mm]below:$u_3$}] (u3) at (0, -\d*9) {};

\node[vertex, label=right:$w_1$] (w1) at (\b*3, -\d*3) {};
\node[vertex, label=right:$w_2$] (w2) at (\b*3, -\d*6) {};
\node[vertex, label=right:$w_3$] (w3) at (\b*3, -\d*9) {};

\draw [ultra thick, gray] (u1) -- node[edgelabel, near start] {1} node[edgelabel, near end] {3} (w1);
\draw [ultra thick, gray!50] (u1) -- node[edgelabel, near start] {3} node[edgelabel, near end] {1} (w2);
\draw [ultra thick] (u1) -- node[edgelabel, near start] {2} node[edgelabel, near end] {2} (w3);

\draw [ultra thick, , gray!50] (u2) -- node[edgelabel, near start] {3} node[edgelabel, near end] {1} (w1);
\draw [ultra thick] (u2) -- node[edgelabel, near start] {2} node[edgelabel, near end] {2} (w2);
\draw [ultra thick, gray] (u2) -- node[edgelabel, near start] {1} node[edgelabel, near end] {3} (w3);

\draw [ultra thick] (u3) -- node[edgelabel, near start] {2} node[edgelabel, near end] {2} (w1);
\draw [ultra thick, gray] (u3) -- node[edgelabel, near start] {1} node[edgelabel, near end] {3} (w2);
\draw [ultra thick, gray!50] (u3) -- node[edgelabel, near start] {3} node[edgelabel, near end] {1} (w3);

\path [->,draw,ultra thick] ($ (u1) !.15! ($ (u1) + (-4/3*\d,2/4*\b) $)$) to [out=45,in=90, distance=4mm] ($ (u1) !.04! (w1) $);
\path [->,draw,ultra thick] ($ (u1) !.25! ($ (u1) + (-4/3*\d,2/4*\b) $)$) to [out=45,in=45, distance=7mm] ($ (u1) !.07! (w2) $);
\path [->,draw,ultra thick] ($ (u2) !.25! ($ (u2) + (-4/3*\d,2/4*\b) $)$) to [out=45,in=135, distance=4mm] ($ (u2) !.1! (w1) $);
\path [->,draw,ultra thick] ($ (u2) !.15! ($ (u2) + (-4/3*\d,2/4*\b) $)$) to [out=45,in=90, distance=4mm] ($ (u2) !.04! (w2) $);
\path [->,draw,ultra thick] ($ (u3) !.15! ($ (u3) + (-4/3*\d,2/4*\b) $)$) to [out=45,in=180, distance=2mm] ($ (u3) !.04! (w1) $);
\path [->,draw,ultra thick] ($ (u3) !.25! ($ (u3) + (-4/3*\d,2/4*\b) $)$) to [out=45,in=135, distance=6mm] ($ (u3) !.13! (w2) $);

\path [double,draw,ultra thick] ($ (u1) !.05! (w3) $) to [out=-135,in=-135, distance=6mm] ($ (u1) !.25! ($ (u1) + (-4/3*\d,2/4*\b) $)$);
\path [double,draw,ultra thick] ($ (u2) !.05! (w3) $) to [out=-117,in=-135, distance=6mm] ($ (u2) !.25! ($ (u2) + (-4/3*\d,2/4*\b) $)$);
\path [double,draw,ultra thick] ($ (u3) !.05! (w3) $) to [out=-90,in=-135, distance=6mm] ($ (u3) !.25! ($ (u3) + (-4/3*\d,2/4*\b) $)$);

\foreach \from in {u1, u2, u3}
   \draw [ultra thick, dotted] (\from) to[out=130,in=30, distance=0.8cm ] ($ (\from) + (-4/3*\d,0) $);
\end{tikzpicture}
\caption{A variable gadget to the left and a clause gadget to the right. Interconnecting edges are dotted. The arrows point to the preferred edge, while double lines denote incomparability.}
\label{fi:np_super}
\end{figure}

\begin{claim}
	If there is a truth assignment to $B$, then there is a super stable matching in~$G$.
\end{claim}

In each variable gadget belonging to a true variable, $\left\{tx,f\bar{x} \right\}$ is chosen, whereas all gadgets corresponding to a false variable contribute matching $\left\{fx,t\bar{x} \right\}$. In each clause, there is at least one true literal.  The vertex representing the appearance of this literal is matched to $w_3$ in the clause gadget, while the remaining four vertices are coupled up in such a way that no edge inside of the gadget blocks. This is possible, because $\left\{u_1w_3, u_2w_2, u_3w_1\right\}$, $\left\{u_1w_1, u_2w_3, u_3w_2\right\}$, and $\left\{u_1w_2, u_2w_1, u_3w_3\right\}$ are all stable matchings. The reason why the literal satisfying the clause was chosen to be matched to $w_3$ is that its the matching edge in the variable gadget is strictly preferred to its interconnecting edge, and thus it does not block~$M$. Due to the strict preferences inside gadgets, it is easy to check that no other edge blocks the constructed matching.

\begin{claim}
	If there is a super stable matching $M$ in $G$, then there is a truth assignment to~$B$.
\end{claim}

If either $t$ or $f$ is unmatched in $M$, then at least one of their $x$ and $\bar{x}$ vertices is either unmatched or it is matched along an interconnecting edge. In both cases, this vertex has a blocking edge leading to the unmatched $t$ or $f$. With this we have already shown three statements: \begin{inparaenum} 
    \item for each variable gadget, either $\left\{tx,f\bar{x} \right\} \in M$ or $\left\{fx,t\bar{x} \right\} \in M$;
    \item no interconnecting edge is in~$M$;
    \item $M$ is perfect in each clause gadget.
\end{inparaenum} In each clause gadget, exactly two $u$-vertices are matched to partners strictly preferred to their interconnecting edge. Therefore, each clause gadget has exactly one interconnecting edge that is incomparable to the edge in $M$ at the clause gadget. In order to ensure stability, this edge must be dominated by $M$ at its variable gadget. This only happens if the corresponding literal is satisfied in the truth assignment. With this we have proved that each clause is satisfied.
\end{proof}

\end{document}